\documentclass[english,a4paper,12pt]{article}
\usepackage[utf8]{inputenc}
\usepackage[T1]{fontenc}
\usepackage{titling}

\usepackage{amsmath, amssymb}
\usepackage{pgf}
\usepackage{amsfonts}
\usepackage{dsfont}
\usepackage{mathrsfs}
\usepackage{mathabx}
\usepackage{stmaryrd}
\usepackage{makeidx}
\usepackage{amsbsy}
\usepackage{amsthm}
\usepackage{color}
\usepackage{fullpage}
\usepackage{enumitem}
\usepackage[vcentermath]{youngtab}
\usepackage{marginnote} %%% margin note
\usepackage{mdframed}	%%% implement corrections
\newmdenv[topline=false, bottomline=false, skipabove=\topsep, skipbelow=\topsep]{siderules}
\newtheorem{theorem}{Theorem}
\newtheorem{corollary}{Corollary}

\newtheorem{proposition}{Proposition}
\newtheorem{lemma}{Lemma}

\newtheorem{definition}{Definition}

\def\sH{\mathscr{H}}

\def\sK{\mathscr{K}}

\def\sP{\mathscr{P}}

\def\cA{{\mathfrak A}}
\def\cB{{\mathfrak B}}

\def\cM{{\mathfrak M}}
\def\cN{{\mathfrak N}}

\newcommand{\ca}[1]{{\cal #1}}
\newcommand{\ben}{\begin{equation}}
\newcommand{\een}{\end{equation}}
\def\bena{\begin{eqnarray}}
\def\eena{\end{eqnarray}}

\def\cA{{\ca A}}
\def\cB{{\ca B}}

\def\cF{{\ca F}}

\def\cH{{\ca H}}

\def\cM{{\ca M}}
\def\cN{{\ca N}}

\renewcommand{\H}{\mathscr{H}}
\newcommand{\F}{\mathcal{F}}

\def\1{{\mathds{1}}}

\def\Hom{{\mathrm{Hom}}}

\newcommand{\dd}{{\rm d}}
\newcommand{\tr}{\operatorname{Tr}}

\renewcommand{\log}{\operatorname{ln}}

\renewcommand{\epsilon}{\varepsilon}

\newcommand{\RR}{\mathbb{R}}
\newcommand{\CC}{\mathbb{C}}

\begin{document}
\title{Variational approach to relative entropies (with application to QFT)
}

	\author{Stefan Hollands$^{1}$\thanks{\tt stefan.hollands@uni-leipzig.de}\\
	{\it $^1$ ITP, Universit\" at Leipzig, MPI-MiS Leipzig and KITP, Santa Barbara}
	}

\date{\today}
	
\maketitle

\begin{abstract}
We  define a new divergence of von Neumann algebras using a variational expression
that is similar in nature to Kosaki's formula for the relative entropy. Our divergence 
satisfies the usual desirable properties, upper bounds the sandwiched Renyi entropy and reduces 
to the fidelity in a limit. As an illustration, we use the 
formula in quantum field theory to compute our divergence between the vacuum in a bipartite system and 
an ``orbifolded'' -- in the sense of conditional expectation -- system in terms of the Jones index. We take the opportunity to 
point out entropic certainty relation for arbitrary von Neumann subalgebras of a factor related to the relative entropy. This certainty 
relation has an equivalent formulation in terms of error correcting codes.
\end{abstract}

\section{Introduction}

The relative entropy between two density operators $\rho, \sigma$, defined as 
\ben\label{eq:Srel}
S(\rho| \sigma) = \tr[\rho (\log \rho - \log \sigma)], 
\een
is an asymptotic measure of their distinguishability. Classically, $e^{-NS(\{p_i\} | \{q_i\})}$ approaches for large $N$ 
the probability for a sample of size $N$ of letters, distributed according to the 
true distribution $\{p_i\}$, when calculated according to an incorrect guess $\{q_i\}$. In the non-commutative setting, the relative entropy has been generalized to von Neumann algebras of arbitrary type by Araki \cite{Araki1,Araki2} using relative modular hamiltonians. 

By far the most fundamental property of $S$ -- from which in fact essentially all others follow -- is its monotonicity under a channel. A channel between von Neumann algebras is a completely positive normal linear map, i.e. roughly an arbitrary combination of (i) a unitary time evolution of the density matrix, (ii) a von Neumann measurement followed by post-selection, (iii) forgetting part of the system (partial trace). The fundamental property is that if $T$ is such a channel and its application to a density matrix is $T[\rho]$ (Schr\" odinger picture\footnote{In the main text we will 
think of $T$ in the Heisenberg picture, i.e. acting on the algebra of observables. Then $\rho(a)=\tr(a\rho)$ is 
identified with a functional on the algebra and $\rho[T]$ corresponds to $\rho \circ T$.}), then always \cite{Uhlmann1977} 
\ben\label{eq:dpi}
S(\rho | \sigma) \ge S( T[\rho] | T[\sigma]). 
\een
In quantum information theory, $T$ is related to data processing, so \eqref{eq:dpi} is sometimes called the data-processing inequality (DPI).

$S$ plays an important role when characterizing the entanglement between subsystems.
Over the years, several generalizations of $S$ with different operational meaning 
have therefore been given, see e.g. \cite{Petz1993}. One such 
generalization is the 1-parameter family of ``sandwiched relative Renyi divergences (entropies)''  $D_s$ proposed by \cite{mueller}.
They interpolate between $S$ and the fidelity $F$, have an operational meaning, and in fact play a major role in recent proofs of improved DPIs for $S$, see \cite{paperI, Junge}. 

The purpose of this note is to point out related variational expression, $\Phi_s$, [eq. \eqref{eq:cst}] 
inspired by a corresponding characterization of $S$ due to Kosaki \cite{Kosaki}. Our formula makes sense for arbitrary 
von Neumann algebras.\footnote{While general von Neumann algebras are not standard in Quantum Information Theory, they are important in other physical applications. For example, in quantum field theory, type $III$ factor are relevant \cite{buchholz}.} It is an upper bound for the sandwiched 
relative Renyi entropies, $D_s$, and it has an interpolating character involving the fidelity and it reduces to that in a limit.  Just as in the case of $S$, the formula is typically not suitable for calculating but can be useful for generalizations (e.g. to $C^*$-algebras or even algebras of unbounded operators), proofs or inequalities. In fact, as we will see, essentially all interesting properties of $\Phi_s$ are simple corollaries of our variational formulas.

One example for this is 
the data processing inequality for $\Phi_s$. 
As another example, we give an application of the 
formula in quantum field theory (QFT). We consider a Haag-Kastler QFT $\cF$ and a subtheory $\cA$, so $\cA \subset \cF$. 
If $A_n, B_n$ are disjoint regions separated by a corridor of size $\sim 1/n$ we can consider 
a conditional expectation ``$E_{A_n} \otimes E_{B_n}$'' projecting $\cF(A_n) \vee \cF(B_n)$ to  $\cA(A_n) \vee \cA(B_n)$.  
The partial state of the vacuum with respect to the subsystem $\cF(A_n) \vee \cF(B_n)$ is called $\omega_\Omega$.
 We show [thm. \ref{thm:1}]
\ben
\lim_{n \to \infty} \Phi_s(\omega_\Omega | E_{A_n} \otimes E_{B_n}[\omega_\Omega]) = \log [\cF:\cA],
\een
which yields a formula \eqref{eq:Fform} for $F$ (fidelity) as a limiting case. Here $[\cF:\cA]$ is the Jones index \cite{Jones,Kosaki1}, whose values 
are restricted to $\{ 4 \cos^2(\pi/n) : n=3,4,\dots\} \cup [4,\infty]$. 
An example is a subtheory $\cA \subset \cF$
of charge neutral operators under a finite gauge group $G$, in which case $[\cF:\cA]=|G|$.\footnote{It has recently been proposed \cite{faulkner}
that the setup of inclusions with conditional expectation may be a model for holography, wherein $\cA,\cF$ correspond to the bulk respectively boundary theory. In such a setting relative entropies between 
$\omega_\Omega$ and $E_{A_n} \otimes E_{B_n}[\omega_\Omega]$ are related to area terms.} Similar results can be obtained in analogous settings
in higher dimensions. 

We also point out a dual 
result for the inclusion $\cF' \subset \cA'$ and the dual conditional expectations $E_n'$ in the case of the fidelity. 
This last result is a consequence of an 
``entropic (un)certainty relation'' (for a review see \cite{Winter}), given in cor.s \ref{cor:3}, \ref{cor:5}, which generalize a result by \cite{pontello} to 
Renyi entropies and general types of Neumann algebras. 
A noteworthy special case of cor. \ref{cor:5} is the following. Consider an inclusion $\cM \supset \cN$, with $\cM$ a factor and $E:\cM \to \cN$ the corresponding conditional expectation with dual conditional expectatation $E':\cN' \to \cM'$. Then we have\footnote{Here $E[\omega_\psi]$ is the 
dual action of the conditional expectation on the partial state (Schr\" odinger picture). In the main text, we write this as $\omega_\psi \circ E$.}
\ben
F_\cM(\omega_\psi | E[\omega_\psi]) \cdot F_{\cN'}(\omega_\psi' | E'[\omega_\psi']) \ge \frac{1}{\sqrt{[\cM:\cN]}}.
\een
Here, $|\psi\rangle$ is a pure state, $\omega_\psi$ the corresponding partial state (density matrix) on $\cM$ and 
$\omega'_\psi$ that on $\cN'$. $F$ is the fidelity between two states. Such relations remind one of the Heisenberg uncertainty principle, 
and connections to various entropic (un)certainty relations are indeed known to exist, see e.g. \cite{Winter}. 
We plan to come back to this in the future.

\medskip
\noindent
{\bf Notations and conventions:} Calligraphic letters $\cA, \cM, \dots$ denote von Neumann algebras. Calligraphic letters $\sH, \sK, \dots$ denote linear spaces.  We use the physicist's ``ket''-notation $|\psi\rangle$ for vectors in a Hilbert space. 
The scalar product is written as 
$ \langle \psi | \psi' \rangle$
and is anti-linear in the first entry. The norm of a vector is written simply as
$\| |\psi\rangle \| =: \| \psi \|$. 
Each vector $|\psi\rangle \in \H$ gives rise to a positive definite linear functional on the von Neumann algebra $\cM$
acting on $\sH$ via
\ben
\label{eq:fnc}
\omega_\psi(m) = \langle \psi | m \psi \rangle, \quad m \in \cM.
\een
The commutant of $\cM$ is denoted as $\cM'$ and consists of those bounded operators commuting with all elements of $\cM$.

\section{Von Neumann algebras and relative entropy}\label{sec:2}

\subsection{Relative modular theory and entropy}
Let $(\cM, J, \sP_\cM^\natural, \sH)$ be a von Neumann algebra in standard form acting on a Hilbert space $\sH$, 
with natural cone $\sP^\sharp_\cM$ and modular conjugation $J$ (for an explanation of these terms, 
see \cite{Bratteli,Takesaki} as general references). 
We will use relative modular operators $\Delta_{\psi,\zeta}$ associated with two vectors  
$|\zeta\rangle, |\psi\rangle \in \sH$ in our constructions. Let $|\psi \rangle, |\zeta \rangle \in \sP^\natural$. Then
there is a non-negative, self-adjoint operator $\Delta_{\psi,\zeta}$ characterized by
\ben 
\label{modA}
J \Delta_{\psi,\zeta}^{1/2}\left( a \left| \zeta \right> + \left| \chi \right> \right) = \pi^\cM(\zeta) a^* \left| \psi \right>\,,  
\quad \forall \,\, a \in \cM \,,\,\, \left| \chi \right> \in (1-\pi^{\cM'}(\zeta)) \mathscr{H} .
\een
Here, $\pi^{\cM'}(\psi)$ is the support projection of the vector $|\psi\rangle$, defined as the orthogonal projection onto $\cM |\psi\rangle$.
The non-zero support of $\Delta_{\psi,\zeta}$ is $\pi^\cM(\psi) \pi^{\cM}(\zeta) \mathscr{H}$, and the functions $\Delta_{\psi,\zeta}^z$
are understood via the functional calculus on this support and are defined as $0$ on $1-\pi^\cM(\psi) \pi^{\cM}(\zeta)$. 

According to \cite{Araki1,Araki2}, if the support projections satisfy
$\pi^\cM(\psi) \ge \pi^\cM(\zeta)$, the relative entropy may be defined by 
\ben
S(\zeta | \psi) = -\lim_{\alpha \to 0^+} \frac{\langle \zeta | \Delta^\alpha_{\psi, \zeta} \zeta \rangle-1}{\alpha} ,  
\een
otherwise, it is by definition infinite. The relative entropy may be viewed as a function of the 
functionals $\omega_\psi, \omega_\zeta$ on $\cM$. So one can write instead also $S(\omega_\zeta | \omega_\psi)$ without ambiguity. In the case of the matrix algebra $M_n(\CC)$, where $\omega_\zeta$ and $\omega_\psi$ are identified with density matrices as $\omega_\psi(a) = \tr(a\omega_\psi)$ etc., the relative entropy is the usual expression \eqref{eq:Srel}.

Kosaki \cite{Kosaki} has given the following variational formula for two normalized state
functionals $\omega_\psi, \omega_\zeta$ on $\cM$:
\ben\label{Kosaki}
S(\omega_\zeta | \omega_\psi) = \sup_{n \in {\mathbb N}} \sup_{x: (1/n,\infty) \to \cM} \left\{
\log n - \int_{1/n}^\infty [ \omega_\zeta(x(t)^* x(t)) + t^{-1} \omega_\psi(y(t) y(t)^*)] t^{-1} \dd t
\right\},
\een
where the second supremum is over all step functions $x(t)$
valued in $\cM$ with finite range where $y(t) = 1-x(t)$. \eqref{Kosaki} no longer makes explicit reference 
to modular theory and the dependence on the state functionals (as opposed to vectors) is manifest. 
Some uses of Kosaki's formula are discussed e.g., in \cite{Petz1993}, ch. 5.

\subsection{Conditional expectations, index, and relative entropy}

Let $\cM, \cN$ be two von Neumann algebras. A linear operator $T:\cM \to \cN$ is called a channel 
if it is ultra-weakly continuous (``normal''), unital $T(1)=1$, and completely positive. 
The latter means that the induced mapping $T \otimes id_n: \cM \otimes M_n \to \cN \otimes M_n$, 
with $M_n$ the full matrix algebra of rank $n$, maps non-negative elements to non-negative elements. 
In particular $T(m^* m)$ is a self-adjoint operator in $\cN$ with non-negative spectrum.

If $\cN \subset \cM$ is a von Neumann sub-algebra, then a quantum channel $E: \cM \to \cN$ is called 
a conditional expectation if 
\ben
E(n_1 m n_2)= n_1 E(m) n_2
\een
for $m \in \cM, n_i \in \cN$. The space of such conditional expectations is called $C(\cM, \cN)$. 
A faithful normal operator valued weight is an unbounded and unnormalized positive linear map $N: \cM \to \cN$
with the same bimodule property and with dense domain $\cM_+$ ($=$ non-negative elements of $\cM$) 
 \cite{haagerup}. The space of such operator-valued weights is denoted $P(\cM, \cN)$, 
and clearly $C(\cM, \cN)$ is a subset thereof. Both $C(\cM,\cN)$ and $P(\cM,\cN)$ may be empty.  

Let $\cM$ be a factor. If there exists $E \in C(\cM,\cN)$, then the best constant $\lambda>0$ such that 
\ben
\label{eq:pp}
E(m^*m) \ge \lambda^{-1} m^*m \quad \text{for all $m \in \cM$}
\een
is called $ind(E)$, the index of $E$. 
If there is any conditional expectation at all, then there is one for which $\lambda$ is minimal \cite{Hiai}. 
This $\lambda = [\cM:\cN]$ is the Jones-Kosaki index of the inclusion  \cite{Jones,Kosaki1,popa}. 

Haagerup  \cite{haagerup} has established a canonical correspondence 
$N \in P(\cM,\cN) \leftrightarrow N^{-1} \in P(\cN', \cM')$ satisfying $(N^{-1})^{-1} = N, (N_1 \circ N_2)^{-1} = N_2^{-1} \circ N_1^{-1}$. 
One can connect this to the notion of a ``spatial derivative'' \cite{connes}. To this end, let $\cM$ be a von Neumann algebra acting on $\H$, 
let $|\zeta \rangle, |\psi\rangle \in \H$. Applying \eqref{eq:fnc} to $\cM$ and the commutant $\cM'$, we get state functionals 
$\omega_\zeta'$ respectively $\omega_\psi$ on $\cM'$ respectively $\cM$. Now the functional $\omega_\psi: \cM \to \CC$ is a
special case of a conditional expectation, so the dual conditional expectation $\omega_\psi^{-1}$ is in $P(B(\H), \cM')$. 
Thus, $\omega'_\zeta \circ \omega_\psi^{-1}$ is a weight on $B(\H)$. Such a weight defines a densely defined positive 
definite (sesqulinear) quadratic form on $\H$ by
\ben
q_{\psi,\zeta}(\phi_1, \phi_2) = \omega'_\zeta \circ \omega_\psi^{-1}(|\phi_2\rangle\langle \phi_1|), 
\een
and the operator $T$ on $\H$ representing $q_{\psi,\zeta}$ is called the ``spatial derivative'', $\Delta_{\cM}(\omega'_\zeta / \omega_\psi)$. 
It can be seen to only depend on the functionals $\omega_\zeta'$ respectively $\omega_\psi$ on $\cM'$ respectively $\cM$.
$\Delta_{\cM}(\omega'_\zeta / \omega_\psi)$ equals the relative modular operator $\Delta_{\cM;\zeta,\psi}$ in case $|\psi\rangle \in {\mathscr P}_\cM$.
It follows that if $|\zeta\rangle$ is in the form domain of $\log \Delta_{\cM}(\omega'_\zeta / \omega_\psi)$, then the relative entropy may also be written 
as 
$
S(\zeta |\psi) = \langle \zeta | \log \Delta_{\cM}(\omega'_\zeta / \omega_\psi) \zeta\rangle.
$
This representation and the structures established by \cite{connes,haagerup} have an immediate corollary for a conditional expectation 
$E: \cM \to \cN$. First, by \cite{connes}, thm. 9,
the spatial derivative has the duality property 
\ben
\Delta_{\cM}(\omega'_\zeta / \omega_\psi) = \Delta_{\cM'}(\omega_\psi / \omega'_\zeta)^{-1}.
\een
Furthermore, 
$
\omega'_\psi \circ (\omega_\psi \circ E)^{-1} = (\omega'_\psi \circ E^{-1}) \circ \omega_\psi^{-1}, 
$
so \cite{Kosaki1}
\ben
\Delta_{\cM}(\omega'_\psi / \omega_\psi \circ E) =  \Delta_{\cN}(\omega'_\psi \circ E^{-1}/ \omega_\psi) 
= \Delta_{\cN'}(\omega_\psi/\omega'_\psi \circ E^{-1})^{-1}.
\een
Taking a log and the expectation value with respect to the vector $|\psi\rangle$ then gives: 
\ben
S_{\cM}(\omega_\psi | \omega_\psi \circ E) + S_{\cN'}(\omega_\psi' | \omega_\psi' \circ E^{-1}) = 0.
\een
Note that $E^{-1}$ is not normalized unless $E=id$. If $\cM$ is a factor such that 
$ind(E)=\lambda < \infty$ is finite, then it can be shown from \eqref{eq:pp} that $1$ is in the domain of $E^{-1}$ and $\lambda 1 = E^{-1}(1)$. Therefore 
\ben
E' = \lambda^{-1}E^{-1}
\een
is a (normalized) conditional expectation $E' \in C(\cN', \cM')$ \cite{Kosaki1}. In fact, if $E$ is minimal, then also $E'$ is. 
Using the standard scaling properties of the relative entropy thereby gives the following trivial corollary 
which generalizes \cite{pontello} who have considered by an explicit method 
the special case of finite dimensional type I von Neumann algebras:
\begin{corollary}
\label{cor:3}
Let $\cN \subset \cM$ be a von Neumann subalgebra of a von Neumann factor $\cM$ with finite index $[\cM:\cN]<\infty$, acting on a Hilbert space $\H$. Assume that $E \in C(\cM,\cN)$ is the minimal conditional expectation, $E' \in C(\cN', \cM')$ the dual minimal conditional expectation. 
For $|\psi\rangle \in \H$, we have
\ben\label{cert}
S_{\cM}(\omega_\psi | \omega_\psi \circ E) + S_{\cN'}(\omega_\psi' | \omega_\psi' \circ E') = \log [\cM:\cN].
\een
(Note that $\omega_\psi'$ in the second expression means the functional \eqref{eq:fnc} on $\cN'$ etc.)
\end{corollary}
Results of a similar flavor have also been given by \cite{Xu1}. Very interesting physical applications of the above ``certainty relation'' \eqref{cert} involving
Wilson- and `t Hooft operators in 4 dimensional quantum Yang-Mills theory have recently been pointed out by \cite{pontello,Casini1}. In such a 
situation the algebras are expected to be of type III \cite{buchholz}.

Then, the minimal conditional expectation $E$ and its dual $E'$ can be described more explicitly using Q-systems \cite{longo2}, see app. \ref{sec:app}. In this framework, $\cM$ is generated by $\cN$ together with a single operator, $v$, 
and $\cN'$ is generated by $\cM'$ together with a single operator, $v'$. The operators $w=j_\cN(v') \in \cN$, $w'=j_\cM(v) \in \cM'$ and the ``canonical'' endomorphsms
\ben
\gamma=j_\cN j_\cM : \cM \to \cN, \quad \gamma' = j_\cM j_\cN: \cN' \to \cM'
\een
can be defined  (here $d=[\cM:\cN]^{1/2}$), where $j_\cN(n)=J_\cN n J_\cN$ and $J_\cN$ is the modular conjugation\footnote{With respect to a
fixed natural cone ${\mathscr P}^\sharp_\cN$.} of $\cN$, etc. The expectations $E,E'$ are then given by
\ben
E(m)=\frac{1}{d} w^* \gamma(m) w, \quad E'(n')= \frac{1}{d}  w^{\prime *} \gamma'(n') w^\prime.
\een
Another property is that $J_\cM v'= v' J_\cN, J_\cM v= v J_\cN$.

The operator $v'$ is closely related to the idea of quantum error correcting codes as described by \cite{faulkner}: 
For the sake of easier comparison, define 
\ben
V:= v'/\sqrt{d}, \quad V':= v/\sqrt{d}, 
\een
with the normalizations made such that $V, V'$ are isometries.
For any $|\psi\rangle, |\zeta\rangle \in \H$ we have the implications
\ben\label{eq:code}
\begin{cases}
\omega_{\zeta}|_{\cN'} = \omega_{\psi}|_{\cN'} \Longrightarrow & \omega_{V\zeta}|_{\cM'} = \omega_{V\psi}|_{\cM'}\\
\omega_{\zeta}|_{\cN} = \omega_{\psi}|_{\cN} \ \  \Longrightarrow & \omega_{V\zeta}|_{\cM} = \omega_{V\psi}|_{\cM}, 
\end{cases}
\een
so $\cM$ is ``standardly c-reconstructible'' from $\cN$ in the terminology \cite{faulkner}. In the context of holography, $\cN$ would be a bulk 
observable algebra, $\cM$ a corresponding CFT algebra and the subspace $V\H \subset \H$ the ``code subspace''. Dually, the operator
$V'$ is used in a similar way to ``standardly c-reconstruct'' $\cN'$ from $\cM'$, with similar relations. While the existence and properties 
of the operator $V$ are equivalent to the existence of some conditional expectation $E:\cM \to \cN$ alone \cite{faulkner}, thm. 7, the existence of 
the operator $V'$ for the dual code does not follow from these results but requires a finite index (and minimal conditional expectation). 

These facts can be used to give an ``error correction version'' of the certainty relation expressed by cor. \ref{cor:3}.  We simply observe the equalities
\ben
E(m) = \frac{1}{d} w^* \gamma(m) w =  \frac{1}{d} j_\cN(v^{\prime *}) j_\cN j_\cM(m) j_\cN(v^{\prime }) = 
\frac{1}{d} J_\cN v^{\prime *} J_\cM m J_\cM v^{\prime } J_\cN = V^* m V
\een
for $m \in \cM$.
Dually, we get $E'(n') = V^{\prime *} n' V'$ for $n' \in \cN'$. This gives in view of cor. \ref{cor:3}:

\begin{corollary} [Error correcting code version]\label{cor:ecc}
Let $\cM \supset \cN$ be an inclusion of type III von Neumann factors with finite index and let $|\psi\rangle \in \H$. Let $V$ be 
a code operator as in \eqref{eq:code} and $V'$ the dual code operator. Then 
\ben
S_{\cM}( \omega_\psi  | \omega_{V\psi} ) + S_{\cN'}(\omega_\psi | \omega_{V'\psi}) = \log [\cM:\cN]. 
\een
\end{corollary}

\subsection{Sandwiched Renyi divergence}
A family of entropy functionals for von Neumann algebras extrapolating the relative entropy are the ``sandwiched Renyi divergences (entropies)'' \cite{mueller}.
In the general von Neumann algebra setting, they can be defined in terms of certain $L_p$ norms.
These weighted $L_p$ spaces were defined by \cite{AM} relative to a fixed cyclic and separating vector $|\psi\rangle \in \H$ in the a natural cone of a standard representation of a von Neumann algebra $\cM$. 

For $1 \le p \le 2$, $L_p(\cM, \psi)$ is defined as the completion of $\sH$ with respect to the following norm:
\ben\label{eq:pnorm}
\|\zeta\|_{p,\psi} = \inf \{  \|\Delta_{\phi, \psi}^{(1/2)-(1/p)}\zeta\| : \|\phi\|=1, \pi^\cM(\phi) \ge \pi^\cM(\psi)=1 \}.
\een
 The generalization to non-faithful state functionals 
$\omega_\psi$, whose representing vector $|\psi \rangle$ is not separating, is given in\footnote{A related approach to non-commutative $L_p$-norms is 
\cite{JencovaLp1}.} \cite{Berta2}, modulo certain technical details related to the H\" older inequality. This has been proven in 
the separating case by \cite{AM} and connects the above norms to those for index $p \ge 2$. In this paper, we restrict 
to the range $1\le p \le 2$, however. 

\begin{definition}
Let $\cM$ be a von Neumann algebra in standard form acting on $\cH$. 
The ``sandwiched Renyi divergences'' \cite{mueller} $D_s, s \in (1/2,1) \cup (1,\infty)$, are defined by 
\ben
D_s(\omega_\zeta | \omega_\psi) = (s-1)^{-1} \, \log \| \zeta \|_{2s,\psi,\cM'}^{2s}
\een
 with norm taken relative to $\cM'$.
 \end{definition}
 
 The sandwiched Renyi divergences extrapolate the relative entropy which can be recovered as the 
 limit $s \to 1^-$. At the other end, for $s \to 1/2^+$, one recovers the log fidelity. In fact, 
 the $L_1$ norm relative to $\cM$ is related to the fidelity \cite{UhlmannFidelity,alberti} relative to $\cM'$ by
\ben
\label{eq:Fchar}
\| \zeta \|_{1,\psi,\cM} = \sup\{ |\langle \zeta | a \psi \rangle | : a \in \cM, \|a\|=1 \} = F_{\cM'}(\omega_\zeta, \omega_\psi),
\een
see \cite{paperI}, lem. 3 (1), which generalizes  \cite{AM}, lem. 5.3 when $\psi$ is not necessarily faithful. 
 
 $D_s$ has an operational meaning in terms of hypothesis testing, see \cite{om}. For density matrices $\omega_\zeta, \omega_\psi$ (corresponding in the case of type I factors to state functionals via $\omega_\psi(a)=\tr(a\omega_\psi)$ etc.), the 
 definition gives
 \ben\label{eq:typeI}
D_s(\omega_\zeta | \omega_\psi) = (s-1)^{-1} \, \log \tr (\omega_\psi^{(1-s)/(2s)} \omega_\zeta \omega_\psi^{(1-s)/(2s)})^s.
 \een
Returning to the case of general von Neumann algebras, we recall that $D_s \le S$ by \cite{Berta2}, prop. 4. Cor. \ref{cor:3}
therefore implies
\begin{corollary}\label{cor:5}
For a finite index inclusion $\cN \subset \cM$ with minimal conditional expectation $E:\cM \to \cN$:
\ben\label{cert1}
D_{s}^\cM(\omega_\psi | \omega_\psi \circ E) + D_{s}^{\cN'}(\omega_\psi' | \omega_\psi' \circ E') \le \log [\cM:\cN].
\een
\end{corollary}
A noteworthy special case arises for $s=1/2$:
\ben
F_\cM(\omega_\psi | \omega_\psi \circ E) \cdot F_{\cN'}(\omega_\psi' | \omega_\psi' \circ E') \ge \frac{1}{\sqrt{[\cM:\cN]}}.
\een
There are also evident error correcting code formulations of this analogous to cor. \ref{cor:ecc}.
 
\section{Variational formulas}

Here we point out a variational formula related to the $L_p$ norms in the range $p \in (1,2)$
similar to Kosaki's 
formula \cite{Kosaki} for the relative entropy. First we assume $|\zeta\rangle$ to be separating for $\cM$, 
hence cyclic for $\cM'$. Similarly as in \cite{Petz1993}, lem. 5.9, we can first argue that
\ben\label{eq:quadint}
\langle \Delta^{-1}(\Delta^{-1}+t)^{-1} \zeta | \zeta \rangle
= \inf\{
\|x \zeta\|^2 + t^{-1} \|\Delta^{-1/2} y \zeta\|^2 : x,y \in \cM', x+y=1
\},
\een
with $\Delta^{-1}=\Delta_{\phi, \psi;\cM}^{-1}=\Delta_{\psi,\phi;\cM'}$ and $t>0$, noting that $y|\zeta\rangle \in {\mathscr D}(\Delta^{-1/2})$
when $y \in \cM'$. Then combining the well-known formula
\ben\label{eq:lambdaa}
\lambda^\alpha = \frac{\sin(\pi \alpha)}{\pi} \int_0^\infty \frac{\lambda }{t+\lambda} t^{\alpha-1} \dd t 
\een
when $\lambda>0, \alpha \in (0,1)$, with  \cite{Petz1993}, prop. 5.10 gives 
\ben
\| \Delta_{\phi,\psi}^{-\alpha/2} \zeta \|^2
=
\frac{\sin(\pi \alpha)}{\pi}
\inf_{x:\RR_+ \to \cM'} \int_0^\infty [\|x(t) \zeta\|^2 + t^{-1} \| \Delta_{\phi,\psi}^{-1/2} y(t) \zeta \|^2] t^{\alpha-1} \dd t, 
\een
where the infimum is taken over all step functions $x:[0,\infty] \to \cM'$ with finite range and $x(t)=1$ for sufficiently small $t>0$
and $x(t)=0$ for sufficiently large $t$, and $y(t)=1-x(t)$. Now taking the infimum as in the definition of the $L_p$ norm and using the 
definition of the $L_1$-norm yields for $\alpha=2/p-1 \in (0,1)$:
\ben\label{eq:interp1}
\| \zeta \|^2_{p,\psi,\cM}
\ge
-\frac{\sin(2\pi/p)}{\pi}
\inf_{x:\RR_+ \to \cM'} \int_0^\infty [\|x(t) \zeta\|^2 + t^{-1} \|  y(t) \zeta \|^2_{1,\psi,\cM}] t^{2/p-2} \dd t.
\een
The $L_1$ norm relative to $\cM$ is related to the fidelity \cite{UhlmannFidelity,alberti} relative to $\cM'$ by \eqref{eq:Fchar}.
Exchanging the roles of $\cM$ and $\cM'$ then gives:
\begin{proposition}\label{prop:var}
If $1 < p < 2$, and $|\zeta\rangle$ is cyclic for $\cM$, we have the variational formula
\ben
\label{eq:cst}
\| \zeta \|_{p,\psi,\cM'}^2
\ge 
c_p
\inf_{x:\RR_+ \to \cM} \int_0^\infty [\omega_\zeta(x(t)^*x(t))  + t^{-1} F_\cM(y(t) \omega_\zeta y(t)^* | \omega_\psi)^2] t^{-2/p'} \dd t , 
\een
for the $L_p$-norm relative to $\cM',\psi$, 
where $F_\cM$ is the fidelity, 
\ben
c_p=-\frac{\sin(2\pi/p)}{\pi}>0, \quad \frac{1}{p} + \frac{1}{p'}=1, 
\een
$y(t)=1-x(t)$, $x: \RR_+ \to \cM$ a step function as described, 
and where we use the notation $(x\omega x^*)(b)=\omega(x^* a x)$. 
\end{proposition}
Note that all terms on the right side of \eqref{eq:cst} manifestly only depend on the functionals 
$\omega_\zeta, \omega_\psi$ on $\cM$ and not their vector representatives $|\zeta\rangle,|\psi\rangle$. Hence, they
can be defined intrinsically on a $C^*$-algebra as well -- for the fidelity 
this follows from another variational formula \cite{UhlmannFidelity,alberti}.
The proposition might 
hence be a possible starting point of an investigation in the context of $C^*$-algebras. 

Note also that we may always go to the GNS-representation of for $\cM$ in the state $\omega_\zeta$, 
in which the state representer is automatically cyclic for $\cM$, so this assumption may in fact be dropped from the proposition\footnote{
If we go to the GNS-representation of $\omega_\zeta$, $\cM$ may no longer be presented in standard form, so we must 
use the Connes spatial derivative to define the $L_p$-norms as in \cite{Berta2}.}. 

We will now start to investigate the variational formula in its own right. For convenience, we make the following definition ($p=2s$).

\begin{definition}
Let $\cM$ be a von Neumann algebra in standard form acting on $\cH$, $s\in (1/2,1)$. 
The ``generalized fidelity'' is defined by 
\ben
\Phi_s(\omega_\zeta | \omega_\psi) =  \, \inf_{x:\RR_+ \to \cM} \log
\left\{
c_{2s} \int_0^\infty [\omega_\zeta(x(t)^*x(t))  + t^{-1} F(y(t) \omega_\zeta y(t)^* | \omega_\psi)^2] t^{\frac{s-1}{s}} \frac{\dd t}{t}
\right\}^{\frac{s}{s-1}}
\een
 with the infimum and notations as defined in prop. \ref{prop:var}.
 \end{definition}

{\bf Remarks:}
1) The normalizations of $\Phi_s$ are chosen in such a way that 
\ben
\Phi_s \ge D_s 
\een
by prop. \ref{prop:var}.

2) The terminonlogy ``generalized fidelity'' is due to the following observation. Consider $\cM = M_n$ and 
diagonal (normalized) density matrices $\omega_\zeta = diag(p_1, \dots, p_n), \omega_\psi = diag(q_1, \dots, q_n)$. 
We use the abbreviation $F=F(\omega_\zeta | \omega_\psi) = \sum_i \sqrt{p_iq_i}$ for the fidelity. By considering the 
variational expression in the definition of $\Phi_s$ with diagonal $x(t) = diag(x_1(t), \dots, x_n(t))$, one can easily convince oneself that the 
infimum can be reached by approximations of 
\ben
x_i(t) = \sqrt{\frac{q_i}{p_i}} \frac{F}{t+1}
\een
by step functions. Inserting this into the variational formula one gets $\Phi_s \ge -\frac{s}{1-s}\log F^2$. We will see below that an 
inequality of this type with a worse constant is true generally. On the other hand, 
as we will also see below, we always have the reverse inequality which implies that $\Phi_s =-\frac{s}{1-s}\log F^2$ in the present case. This 
becomes (minus log of) the squared fidelity when $s=1/2$. 

3) The properties shown below indicate that $\Phi_s$ has most of the desired properties of a divergence. To the best of our knowledge $\Phi_s$ 
is a new generalization of the log fidelity. 

\medskip 
\noindent
We now investigate some properties of $\Phi_s$. First, consider $|\zeta_1\rangle, |\zeta_2\rangle$ 
such that $\omega_{\zeta_1} \le \omega_{\zeta_2}$ in the sense of functionals on the von Neumann algebra $\cM$. 
It is well-known that such a condition 
implies the existence of $a' \in \cM'$ such that $|\zeta_1\rangle = a'|\zeta_2\rangle$ and $\|a'\| \le 1$. Then, 
\eqref{eq:Fchar} immediately gives: 
\ben
\begin{split}
F_\cM(y\omega_{\zeta_1}y^*, \omega_\psi) &= \sup\{
|\langle y \zeta_1 | b' \psi \rangle : b' \in \cM', \| b' \| =1\}\\
&= \sup\{
|\langle y a' \zeta_2 | b' \psi \rangle : b' \in \cM', \| b' \| =1\}\\
&= \sup\{
|\langle y \zeta_2 | a^{\prime *} b' \psi \rangle : b' \in \cM', \| b' \| =1\}\\
&\le \sup\{
|\langle y \zeta_2 | c' \psi \rangle : c' \in \cM', \| c' \| =1\}\\
&= F_\cM(y\omega_{\zeta_2}y^*, \omega_\psi)
\end{split}
\een
for any $y \in \cM$, since $\|a^{\prime *} b'\| \le 1$ so the sup in the fourth line is over a larger set.  
But then the variational formula also gives without difficulty $\| \zeta_1 \|_{p,\psi} \le \| \zeta_2 \|_{p,\psi}$.
This is consistent with the formula \eqref{eq:typeI} in the type I setting because the function $x \mapsto x^s$ is 
operator monotone for $s \in [0,1]$.
Similarly, consider $|\psi_1\rangle, |\psi_2\rangle$ 
such that $\omega_{\psi_1} \le \omega_{\psi_2}$. By the same argument $F(y\omega_{\zeta}y^*, \omega_{\psi_1}) \le F(y\omega_{\zeta}y^*, \omega_{\psi_2})$, and the variational formula thereby gives $\| \zeta \|_{p,\psi_1} \le \| \zeta \|_{p,\psi_2}$. 
In conclusion, we get:

\begin{corollary}
For normal positive functionals on a von Neumann algebra $\omega_{\zeta_1} \le \omega_{\zeta_2}$ and 
$\omega_{\psi_1} \le \omega_{\psi_2}$ we have 
also $\Phi_s(\omega_{\zeta_1} | \omega_{\psi_1}) \ge \Phi_s(\omega_{\zeta_2} | \omega_{\psi_2})$ when $1 > s > 1/2$.
\end{corollary} 

As an application, consider a von Neumann subalgebra $\cN \subset \cM$ together with a conditional expectation 
$E: \cM \to \cN$ and unit vector $|\zeta\rangle$ such that $ind(E)=\lambda<\infty$. Then by definition
$\omega_\zeta \circ E \ge \lambda^{-1} \omega_\zeta$.
The identity $\Phi_s(\omega_\zeta | \lambda^{-1} \omega_\psi) = \Phi_s(\omega_\zeta | \omega_\psi) + \log \lambda$ [cor. \ref{cor:6},3)] and the corollary 
trivially give
\ben\label{Dsbound}
\Phi_s(\omega_\zeta| \omega_\zeta \circ E) \le \log \lambda
\een
because $\Phi_s(\omega_\psi | \omega_\psi)=0$.

As another application of prop. \ref{prop:var}, we can prove the DPI for $\Phi_s$ in the context of properly infinite von Neumann algebras using only properties of the fidelity in the range $1/2 \le s \le 1$, 
without the use of any complex interpolation arguments or modular operators as in \cite{Berta2} or in \cite{frank} in the context of $D_s$. 

\begin{corollary}\label{thm:sDPI}
Let $\cM, \cN$ be properly infinite von Neumann algebras and $T: \cM \to \cN$ a channel. 
Then for two normal state functionals $\omega_\zeta, \omega_\psi$ we have $\Phi_s(\omega_\zeta \circ T| \omega_\psi \circ T)
\le \Phi_s(\omega_\zeta | \omega_\psi)$ for $s \in (1/2,1)$.
\end{corollary}

\begin{proof}
By \cite{landauer}, thm. 2.10, $T$ can be written in Stinespring form $T(b)=v^*\rho(b)v$, where $v \in \cM, v^* v=1, vv^*=q$ 
($q$ a projection) and $\rho:\cN \to \cM$ 
a homomorphism of von Neumann algebras. Then, it is sufficient to prove the theorem separately for the case (i) $T_1(a) = v^*av$ and 
the case (ii) $T_2(b)=\rho(b)$. 

(i) Using \eqref{eq:Fchar} with $\cM'$ in place of $\cM$, 
we have for $y \in \cM$:
\ben
\begin{split}
F_{\cM}(y\omega_{v\zeta}y^* | \omega_{v\psi})  & = \sup \{ 
|\langle yv\zeta | x'v\psi \rangle | : \| x' \|=1, x' \in \cM' 
\}\\
& = \sup \{ 
|\langle yv\zeta | vx'\psi \rangle | : \| x' \|=1, x' \in \cM' 
\}\\
& = \sup \{ 
|\langle v^*yv\zeta | x'\psi \rangle | : \| x' \|=1, x' \in \cM' 
\} \\
& = F_{\cM}( (v^*yv) \omega_{\zeta} (v^*yv)^* | \omega_\psi).
\end{split}
\een
Furthermore, 
\ben
\omega_{v\zeta}( x^* x) = \omega_\zeta(v^* x^* x v) \ge \omega_\zeta((v^*xv)^* v^*xv) .
\een
Then we have in view of prop. \ref{prop:var} ($p=2s$)
\ben
\begin{split}
& c_p \inf_{x: \RR_+ \to \cM} \int_0^\infty [\omega_{v\zeta}( x(t)^* x(t))   + t^{-1} F_\cM( y(t) \omega_{v\zeta} y(t)^* | \omega_{v\psi})^2] t^{-2/p'} \dd t \\
= & c_p \inf_{x: \RR_+ \to \cM} \int_0^\infty [\omega_{v\zeta}( x(t)^* x(t))   + t^{-1} F_\cM( (v^* y(t) v) \omega_{\zeta} (v^* y(t) v)^* | \omega_{\psi})^2] t^{-2/p'} \dd t \\
\ge & c_p
\inf_{x:\RR_+ \to \cM} \int_0^\infty [\omega_\zeta(X(t)^* X(t)) + t^{-1} F_{\cM}( Y(t) \omega_{\zeta} Y(t)^* | \omega_\psi)^2 ] t^{-2/p'} \dd t .
\end{split}
\een
Note that $Y(t)=v^*y(t) v, X(t)=v^* x(t) v$ are particular examples of piecewise constant functions valued in $\cM$ with finite range such that 
$X(t)+Y(t)=1$ and such that $Y(t)=0$ for sufficiently small $t$ and $X(t)=0$ for sufficiently large $t$. Thus, we can make the right side at most smaller by 
taking the infimum over {\em all} such functions. This results in $\Phi_s(\omega_{v\zeta}| \omega_{v\psi}) \le \Phi_s(\omega_{\zeta}| \omega_{\psi})$ using 
the definition of $\Phi_s$.

(ii) We have ($p=2s$)
\ben
\begin{split}
& c_p
\inf_{x:\RR_+ \to \rho(\cN)} \int_0^\infty [ 
\omega_\zeta(x(t)^* x(t)) + t^{-1}
F_{\rho(\cN)}(y(t)\omega_\zeta  y(t)^* | \omega_\psi)^2
] t^{-2/p'} \dd t \\
\ge &  \ c_p
\inf_{X:\RR_+ \to \cM} \int_0^\infty [ 
\omega_\zeta(X(t)^* X(t)) + t^{-1}
F_{\rho(\cN)}(Y(t) \omega_\zeta  Y(t)^* | \omega_\psi)^2
] t^{-2/p'} \dd t \\
\ge & \ c_p
\inf_{X:\RR_+ \to \cM} \int_0^\infty [ 
\omega_\zeta(X(t)^* X(t)) + t^{-1}
F_{\cM}(Y(t) \omega_\zeta  Y(t)^* | \omega_\psi)^2
] t^{-2/p'} \dd t,
\end{split}
\een
where in the first step we took the infimum over the larger set of piecewise constant functions $X$ valued in $\cM$ with finite range
such that $1-X(t)=Y(t)=0$ for sufficiently small $t$ and $X(t)=0$ for sufficiently large $t$. In the second step, we used the monotonicity  
$F_{\rho(\cN)} \ge F_\cM$ since $\rho(\cN)$ is a von Neumann subalgebra of $\cM$, by \eqref{eq:Fchar}. This 
yields $\Phi_s(\omega_{\zeta} \circ \rho| \omega_{\psi} \circ \rho) \le \Phi_s(\omega_{\zeta}| \omega_{\psi})$.
\end{proof}

Applying the DPI to the channel 
$ \cA \to \cA \oplus \cdots \oplus \cA, a \mapsto a \oplus \cdots \oplus a$ and the 
states $\rho = \oplus_i \lambda_i \omega_{\psi_i}, \sigma = \oplus_i \lambda_i \omega_{\zeta_i}$ implies that $\Phi_s$ is 
jointly convex by a standard argument, see e.g. \cite{mueller}, proof of prop. 1, 
\ben
\sum_i \lambda_i \Phi_s(\omega_{\zeta_i} | \omega_{\psi_i}) \ge \Phi_s(\sum_i \lambda_i \omega_{\zeta_i} | \sum_j \lambda_j \omega_{\psi_j})
\een
where the sum is finite and $\lambda_i \ge 0, \sum \lambda_i=1$.  Next, we obtain the following corollary:

\begin{corollary}\label{cor:6}
Let $\cM$ be a von Neumann algebra and $s \in (1/2,1)$.

1) We have for $\|\zeta\|=1$ 
\ben
\label{eq:interp3}
\Phi_s(\omega_\zeta | \omega_\psi) \ge  -\log F(\omega_\zeta | \omega_\psi)^2.  
\een

2) We have for $\|\psi\|=1$
\ben
\label{eq:interp3b}
\Phi_s(\omega_\zeta | \omega_\psi) \le  -\frac{s}{1-s} \log F(\omega_\zeta | \omega_\psi)^2.  
\een

3) $\Phi_s(\omega_\zeta | \lambda \omega_\psi) = \Phi_s(\omega_\zeta | \omega_\psi) - \log \lambda$ for $\lambda > 0$.

4) We have for $\|\psi\|=1=\|\zeta\|$ that $\lim_{s \to (1/2)^+} \Phi_{s} (\omega_\zeta | \omega_\psi)  =  -\log F(\omega_\zeta | \omega_\psi)^2$.

5) $\Phi_s(\omega_\zeta | \omega_\psi) \ge 0$ for $\|\psi\|=1=\|\zeta\|$ with equality iff  $\omega_\zeta = \omega_\psi$.
\end{corollary}

\begin{proof}
For 1), we choose an approximation of 
\ben
x(t) = \frac{F(\omega_\zeta | \omega_\psi)}{t +F(\omega_\zeta | \omega_\psi)} 1
\een
by step functions. Then we apply the variational definition of $\Phi_s$ and the integral formula \eqref{eq:lambdaa} upon which the result follows by a simple calculation.

For 2),  we first use the supremum characterization of the fidelity \eqref{eq:Fchar}, by which 
 have $F(y\omega_\zeta y^*,\omega_\psi)^2\ge |\langle \psi | y \zeta \rangle|^2 = \| P_\psi y \zeta\|^2$, with $P_\psi = |\psi \rangle \langle \psi |$
 a projector. Then ($p=2s$), 
 \ben
 \begin{split}
& c_p
\inf_{x:\RR_+ \to \cM} \int_0^\infty [\omega_\zeta(x(t)^*x(t))  + t^{-1} F(y(t) \omega_\zeta y(t)^* | \omega_\psi)^2] t^{-2/p'} \dd t
\\
\ge & \
c_p
\inf_{x: \RR_+ \to \cM'} \int_0^\infty [\|x(t) \zeta \|^2  + t^{-1} \| P_\psi y(t) \zeta\|^2] t^{-2/p'} \dd t \\
= & \
c_p \int_0^\infty \langle \zeta | P_\psi(t+P_\psi)^{-1}  \zeta\rangle t^{-2/p'} \dd t \\
= & \
c_p \| P_\psi \zeta\|^2  \int_0^\infty (t+1)^{-1}  t^{-2/p'} \dd t 
=  | \langle \zeta | \psi \rangle |^2.
\end{split}
\een
This remains true if we change $|\zeta\rangle \to u' |\zeta\rangle$ for any unitary $u'$ from $\cM'$, thus giving 
 \ben
 \begin{split}
& c_p  \inf_{x:\RR_+ \to \cM} \int_0^\infty [\omega_\zeta(x(t)^*x(t))  + t^{-1} F(y(t) \omega_\zeta y(t)^* | \omega_\psi)^2] t^{-2/p'} \dd t \\
& \ge \sup \{ | \langle u' \zeta | \psi \rangle |^2 : u' \in \cM' \ \ \text{unitary}\} = F(\omega_\zeta | \omega_\psi)^2, 
\end{split}
\een
using a well-known characterization \cite{alberti} of the fidelity in the last step. The rest then follows from the definition of $\Phi_s$.

For 3), we use the homogeneity of the fidelity $F(\lambda y(t)\omega_\psi y(t)^* | \omega_\zeta) = \sqrt{\lambda} F( y(t)\omega_\psi y(t)^* | \omega_\zeta)$
inside the variational formula in the definition of $\Phi_s$ 
and apply a change of variables $t' = t/\lambda$ in the integral. 

Item 4) is a combination of 1) and 2).

Item 5) follows from the properties $F(\omega_\zeta | \omega_\psi) \le 1$,  $F(\omega_\zeta | \omega_\psi)=1$ iff $\omega_\zeta = \omega_\psi$, 
and 1), 2).
\end{proof}

\section{Application to quantum field theory}
\label{sec:3}

Here we consider an application of $\Phi_s$ to quantum field theory inspired by \cite{Xu}. For simplicity and concreteness, 
we consider chiral conformal quantum field theories (CFTs) on a single lightray (real line) or equivalently the circle in the conformally compactified picture. But the arguments are of a rather general nature and would apply with some fairly obvious modifications to general quantum field theories in higher dimensions under appropriate hypotheses.

We assume standard axioms common in algebraic quantum field theory \cite{haag_2}. 
According to this axiom scheme, fulfilled by many examples, a chiral CFT is an 
assignment $I \mapsto \cA(I)$, wherein $I \subset S^1$ is an open interval and 
$\cA(I)$ a von Neumann algebra acting on a fixed Hilbert space, $\sH$. One assumes:

\begin{enumerate}
\item (Isotony)
If $I_1 \subset I_2$ then $\cA(I_1) \subset \cA(I_2)$.

\item (Commutativity)
If $I_1 \cap I_2$ is empty, then $[\cA(I_1) , \cA(I_2)]=\{0\}$.

\item (M\" obius covariance)
There is a strongly continuous unitary representation $U$ on $\sH$ of the M\" obius group $G=\widetilde{SL_2({\mathbb R})/{\mathbb Z}_2}$ which is consistent 
with the standard action of this group the circle by fractional linear transformations, 
in the sense $U(g) \cA(I) U(g)^* = \cA(gI)$ for all $g \in G$.

\item (Positive energy)
The 1-parameter subgroup of rotations has a positive generator $L_0$ under the representation $U$.

\item (Vacuum) 
There is a unique vector $|\Omega\rangle \in \cH$, called the vacuum, which is invariant under all $U(g), g \in G$.

\item (Additivity)
Let $I$ and $I_n$ be intervals such that $I = \cup_n I_n$. Then $\cA(I)=\vee_n \cA(I_n)$ (strong closure).

\end{enumerate}

The special situation we would like to study here are two chiral CFTs $\cA, \cF$ in the above sense such that 
$\cA(I) \subset \cF(I)$ is an inclusion of von Neumann algebras acting on the same Hilbert space $\sH$ 
for any interval $I$, and transforming under the same representation, $U$. A typical example is when 
$\cA$ is the Virasoro net (operator algebras generated by the stress energy tensor) and $\cF$ is an extension
of finite index as classified in \cite{Kawahigashi}. For further 
details on such a setting, see e.g. \cite{longo2,longo3}. We will also assume that the Jones-Kosaki index 
$\lambda \equiv [\cF(I):\cA(I)]$ is finite (hence independent of $I$ \cite{longo3}). 
By \cite{longo2}, lemma 13, this implies that for each $I$ there is a 
conditional expectation $E_I: \cF(I) \to \cA(I)$, satisfying the Pimsner-Popa inequality \eqref{eq:pp}. 
We assume that $E_I$ leaves the vacuum vector invariant, $\omega_\Omega \circ E_I = \omega_\Omega$ for all intervals 
$I$. Furthermore, these conditional expectations must be consistent in the sense 
$E_I |_{\cF(J)} = E_J$ for $J \subset I$. Consider two sets of intervals (identifying $S^1$ with the real line via a stereographic projection):
\ben\label{systems}
A_n = (a, -1/n), \quad B_n = (1/n, b), 
\een
wherein $n$ is a natural number. Thus, $dist(A_n,B_n) = 2/n$ and when $n \to \infty$, the intervals touch each other. 
We consider the von Neumann algebra inclusion 
$\cA(A_n) \vee \cA(B_n) \subset \cF(A_n) \vee \cF(B_n)$, and we let $E_n$ be the conditional expectation 
$\cF(A_n) \vee \cF(B_n) \to \cA(A_n) \vee \cA(B_n)$ such that 
\ben
\label{eq:Endef}
E_n(a_n b_n) = E_{A_n}(a_n) E_{B_n}(b_n) \quad \text{$\forall a_n \in \cF(A_n), b_n \in \cF(B_n)$.}
\een
Thus, $E_n$ only projects out degrees of freedom of the individual parts of the system in \eqref{systems} separately\footnote{Somewhat 
formally $E_n = E_{A_n} \otimes E_{B_n}$, which holds rigorously if the split property holds in the CFT.} In the 
limit as $n \to \infty$ (denoted as $\lim_n$ in the following), these systems touch each other. 
We can show the following 
theorem. 
\begin{theorem}
\label{thm:1}
We have $\lim_n \Phi_s(\omega_\Omega | \omega_\Omega \circ E_n) = \log [ \cF : \cA ]$ for $s \in [1/2,1)$. 
\end{theorem}

We remark that in view of cor. \ref{cor:6}, 4), the limit $s \to 1/2^+$ corresponds to 
\ben
\label{eq:Fform}
\lim_n F(\omega_\Omega|\omega_\Omega \circ E_n) = [\cF : \cA]^{-1/2}.
\een
\begin{proof}
The proof strategy is similar to that of a result by Longo and Xu \cite{Xu} who have considered the relative 
entropy $S$ instead of the divergence $\Phi_s$. As their proof, we make use of the variational 
definition of the divergence $\Phi_s$.

First assume that $1/2<s<1$. We use the notation $d^2  = \lambda \equiv [\cF(I):\cA(I)] < \infty$ which is independent of $I$.
Let $|\psi_n\rangle$ be a vector such that $\omega_{\psi_n} = \omega_\Omega \circ E_n$, as a functional 
on $\cF(A_n) \vee \cF(B_n)$. 

\begin{lemma} \label{lem:1}
There exists a sequence $\{ f_n \} \subset \cF(A_n) \vee \cF(B_n)$ such that $f_n \to 1$ strongly and
\ben \label{anprop}
\lim_n \omega_\Omega(f_n) = 1, \quad \lim_n \omega_\Omega(f_n^* f_n) = 1, \quad
\lim_n \omega_{\psi_n}(f_n^*f_n) = \lambda^{-1}.
\een 
\end{lemma}
\begin{proof}
The proof is given in \cite{Xu}, prop. 4.5. However we rephrase it somewhat in preparation to 
the discussions in the next section. A finite index inclusion $\cN \subset \cM$ of von Neumann factors is characterized 
uniquely by its associated Q-system \cite{longohopf,bischoff} $(x,w,\theta)$, wherein $x,w \in \cN$ obey certain relations 
relative to the endomorphism $\theta$ of $\cN$, see appendix \ref{sec:app}. 

Applying this structure to the inclusions $\cA(A_n) \subset \cF(A_n)$ we get $v_{A_n} \in \cF(A_n)$ and similarly 
for $B_n$. These are fixed uniquely demanding that the corresponding conditional 
expectations $E_{A_n}$ be given by our $\Omega$ preserving conditional expectation $E_{A_n}$ etc. By 
translation-dilation covariance, this implies for example that $v_{A_n} \to v_A$ strongly as $n \to \infty$.
Another standard result in this setting, shown e.g. in \cite{Xu}, lemma 2.9, is that 
$v_{A_n}$ can be ``transported'' to $v_{B_n}$ in the sense that there is a unitary $u_{B_nA_n} \in \cA(a,b) \cap \Hom(\theta_{B_n}, \theta_{A_n})$, 
such that $v_{B_n} = u_{B_nA_n}v_{A_n}$. 
By additivity, we may find a sequence of unitaries 
$a_{n,k} \in \cA(A_n), b_{n,k} \in \cB(B_n)$ 
such that $\sum_{k=1}^{N(n)} b_{n,k}^* a_{n,k}^{} - u_{B_nA_n} \to 0$ as $n\to \infty$, in the strong sense. 
Then, let 
\ben
V_{A_n,k} = \frac{1}{\sqrt{d}} a_{n,k} v_{A_n}  \in \cF(A_n),\quad V_{B_n,k}^* = \frac{1}{\sqrt{d}} v_{B_n}^* b_n^* \in \cF(B_n).
\een
Finally, let 
\ben\label{fdef}
f_n =  \sum_{k=1}^{N(n)} V_{B_n,k}^* V_{A_n,k}^{}.
\een 
Then it follows that $f_n \to d^{-1} v_{B}^* v_{B}^{}=1$ strongly by construction and the relations of Q-systems, see appendix \ref{sec:app}. 
This already implies the first two 
of the claimed limits in \eqref{anprop}. On the other hand,
\ben
\begin{split}
\omega_\Omega\circ E_n(f_n^* f_n) =& \sum_{k,l} \omega_\Omega\circ E_n(V_{A_n,k}^* V_{B_n,k}V_{B_n,l}^* V_{A_n,l}) \\
=&  \sum_{k,l} \omega_\Omega\circ E_n(V_{A_n,k}^* V_{A_n,l} V_{B_n,k}V_{B_n,l}^*) \\
=&  \sum_{k,l} \omega_\Omega(E_{A_n}(V_{A_n,k}^* V_{A_n,l})E_{B_n}(V_{B_n,k}V_{B_n,l}^*)) \\
=&  \sum_{k,l} \omega_\Omega(V_{A_n,k}^* V_{A_n,l}E_{B_n}(V_{B_n,k}V_{B_n,l}^*)) \\
=&  d^{-3} \sum_{k,l} \omega_\Omega(v_{A_n}^*a_{n,k}^* a_{n,l}^{} v_{A_n} b_{n,k}^{} b_{n,l}^*) \\
=&  d^{-3} \sum_{k,l} \omega_\Omega(v_{A_n}^*a_{n,k}^* b_{n,k}^{} a_{n,l}^{} b_{n,l}^* v_{A_n}^{}) \to  d^{-2} 
\end{split}
\een
using commutativity in the first line, the definition of $E_n$ in the second line, 
$E_I |_{\cF(J)} = E_J$ for $J \subset I$ and $\omega_\Omega \circ E_I = \omega_\Omega$
in the third line, identities for a Q-system in the fourth line, commutativity again 
in the fifth line, and $\sum a_{n,k}^* b_{n,k} a_{n,l} b_{n,l}^* \to 1$ strongly and $v_{A_n}^* v_{A_n}^{}=d \cdot 1$ in the last line using again 
properties of the Q-system. 
\end{proof}

Next, we define 
\ben
x_n(t) = 
\begin{cases}
1-\frac{t}{t+\lambda^{-1}} f_n & \text{if $1/k \le t \le k$}\\
1 & \text{if $t>k$}\\
0 & \text{if $t<1/k$}.
\end{cases}
\een
Using the properties \eqref{anprop} of $f_n$, we have for $t \in (1/k, k)$:
\ben
\begin{split}
&\lim_n \omega_\Omega(x_n(t)^* x_n(t)) = \frac{\lambda^{-2}}{(t+\lambda^{-1})^2}\\
&\limsup_n F(y_n(t) \omega_\Omega y_n(t)^* | \omega_{\psi_n}) \le \limsup_n \| y_n(t)^* \psi_n \| = \frac{\lambda^{-1}t^2}{(t+\lambda^{-1})^2},
\end{split}
\een
using in the second line the Cauchy-Schwarz inequality in order to estimate the fidelity characterized through \eqref{eq:Fchar}.
Therefore, for fixed $k$, we have 
\ben
\begin{split}
&\limsup_n \int_{1/k}^k \left[ \omega_\Omega(x_n(t)^* x_n(t)) + t^{-1} F(y_n(t) \omega_\Omega y_n(t)^* | \omega_{\psi_n})^2 \right]
t^{-(2s-1)/s} \dd t\\
\le & \int_{1/k}^k \left[ 
\frac{\lambda^{-2}}{(t+\lambda^{-1})^2} + \frac{\lambda^{-1}t}{(t+\lambda^{-1})^2}
\right] t^{-(2s-1)/s} \dd t \\
= & \ \ c_{2s}^{-1} \lambda^{(s-1)/s} - \frac{s}{1-s} k^{(s-1)/s} - \frac{s}{2s-1} k^{-(2s-1)/s} \\
& + s\sum_{m=1}^\infty (-1)^m \left\{
\left( \frac{\lambda}{k} \right)^m \frac{1}{ms+(1-s)} + \left( \frac{1}{\lambda k} \right)^{m+1} \frac{1}{ms+(2s-1)}
\right\}, 
\end{split}
\een
using the integral \eqref{eq:lambdaa} and the definition of $c_p$ from prop. 1 in the last step. 
The last sum is of order $O(k^{-1})$ uniformly in $s \in [1/2,1]$. On the other hand, 
using the definition of $x_n(t)$ in the range $t<1/k$, we have 
\ben
\begin{split}
&\limsup_n \int_0^{1/k} 
\left[ \omega_\Omega(x_n(t)^* x_n(t)) + t^{-1} F(y_n(t) \omega_\Omega y_n(t)^* | \omega_{\psi_n})^2 \right]
t^{-(2s-1)/s} \dd t\\
= & \int_0^{1/k} t^{-(2s-1)/s} \dd t = \frac{s}{1-s} k^{(s-1)/s} 
\end{split}
\een
while using the definition of $x_n(t)$ in the range $t>k$, we have 
\ben
\begin{split}
&\limsup_n \int_k^\infty 
\left[ \omega_\Omega(x_n(t)^* x_n(t)) + t^{-1} F(y_n(t) \omega_\Omega y_n(t)^* | \omega_{\psi_n})^2 \right]
t^{-(2s-1)/s} \dd t\\
= & \int_k^\infty t^{-(2s-1)/s-1} \dd t = \frac{s}{2s-1} k^{-(2s-1)/s}.
\end{split}
\een
Consequently, when $s=p/2$, the variational expression \eqref{eq:cst} gives us\footnote{Note that the variational 
expression holds by continuity also for strongly continuous families such as $x_n(t)$.} 
\ben
\label{eq:limsup}
\begin{split}
& \limsup_n  c_{2s} \int_0^{\infty} 
\left[ \omega_\Omega(x_n(t)^* x_n(t)) + t^{-1} F(y_n(t) \omega_\Omega y_n(t)^* | \omega_{\psi_n})^2 \right]
t^{-2/(2s)'} \dd t \\
\le &  \ \lambda^{(s-1)/s} + O(k^{-1})
\end{split}
\een
for any $k>0$ where $O(k^{-1})$ is a term bounded in norm by $Ck^{-1}$ uniformly in $s \in [1/2,1]$. 
Letting $k \to \infty$ this term disappears, and then using the definition of $\Phi_s$ and of $\psi_n$ gives
\ben
\liminf_n \Phi_s(\omega_\Omega | \omega_\Omega \circ E_n) \ge \log \lambda.
\een
On the other hand, we have already seen before in \eqref{Dsbound} that $\Phi_s(\omega_\Omega| \omega_\Omega \circ E_n) \le \log \lambda$. The proof of the theorem is therefore complete for the case $1/2<s<1$. 

Now we turn to the limiting case $s \to (1/2)^+$. 
We go back to the proof and investigate the limit as $s \to (1/2)^+$. By inspection it can be seen that in 
order to obtain an expression in \eqref{eq:limsup} not exceeding $\lambda^{(s-1)/s} + O(k^{-1})
 + \epsilon$ for some $\epsilon>0$, we need $n \ge n_0(k,\epsilon)$, where $n_0$ does not depend on $s \in [1/2,1]$. 
 Furthermore, we have argued in the proof that 
 $O(k^{-1})$ is also uniform in $s \in [1/2,1]$. Thus, the limit $s \to 1/2^+$ may be taken and we learn that 
 $F(\omega_\Omega | \omega_{\psi_n})^2 \le \lambda^{-1} + O(k^{-1}) + \epsilon$ when $n \ge n_0(k,\epsilon)$. 
 Thus, $\limsup_n F(\omega_\Omega | \omega_{\psi_n})^2 \le \lambda^{-1}$ and the rest is as before. 
\end{proof}

Cor. \ref{cor:3} for $s=1/2$ gives the following dual formulation of this result when applied to $\cM_n=\cF(A_n) \vee \cF(B_n), 
\cN_n = \cA(A_n) \vee \cA(B_n)$ and $E_n': (\cA(A_n) \vee \cA(B_n))' \to (\F(A_n) \vee \cF(B_n))'$, which is the dual conditional expectation. 
We conclude in view of cor. \ref{cor:6}, 4) that 
\ben
\lim_n F(\omega_\Omega' | \omega_\Omega' \circ E_n') = 1.
\een

\section{Conclusions}
\label{sec:disc}

We end this paper by commenting on the physical significance of the result in sec. \ref{sec:3}. For this it is instructive to have in 
mind the example of a quantum field theory, $\cF$, containing charged fields. These map the vacuum $|\Omega\rangle$
to states with net (flavor) charge. The subset of charge neutral operators is $\cA$. On the full Hilbert space $\cH$ (including charged states), the gauge group $G$ acts by global unitaries which transform the charged fields and leave the vacuum invariant. 
The conditional expectation $E_I: \cF(I) \to \cA(I)$ is the Haar-average over $G$ and projects onto the charge neutral operators (``observables'') in a given region $I$, which is left invariant because gauge transformations commute with translations by the Coleman Mandula theorem. Assuming that $G$ is a finite group with $|G|$ elements, the index is $|G|=[\cF:\cA]$. 

Given two spacelike related regions $A_n$ and $B_n$ separated by a finite corridor of size $\sim 1/n$, the conditional expectation $E_n$ defined by \eqref{eq:Endef} is basically the tensor product $E_{A_n} \otimes E_{B_n}$.
$\Phi_s(\omega_\Omega | \omega_\Omega \circ E_n)$ in a sense accounts for the correlations between 
$A_n$ and $B_n$ that are visible using charge operators only in both subsystem. This interpretation becomes more and more precise when the regions move together. The above intuitive 
argument has been substantiated (in a somewhat heuristic way) in the very lucid paper by \cite{casini}, in the case of the relative entropy $S$ -- such that 
we should use Kosaki's variational formula for $S$ \eqref{Kosaki} instead of the variational definition of $\Phi_s$. 
They first argue using known properties of $S$ in connection with conditional expectations that the mutual information between $A_n$ and $B_n$
in the vacuum state satisfies
\ben
I_\cF(A_n | B_n) - I_\cA(A_n | B_n) = S(\omega_\Omega | \omega_\Omega \circ E_n). 
\een
When $n \to \infty$, it is plausible that the mutual information on the left side 
is dominated by correlations between charge carrying operators localized very near the 
edges where $A_n$ and $B_n$ approach each other. Furthermore, although 
each term in $I_\cF(A_n | B_n) - I_\cA(A_n | B_n)$ is expected 
to diverge, the difference ought to be a finite number related to the order of $G$.
In fact, by investigating more closely the right side of the equation, they argue that $S(\omega_\Omega | \omega_\Omega \circ E_n)$ converges to $\log |G|$ when $n \to \infty$. 

Actually, the core of the argument by \cite{casini} has a similar flavor to ours, in the following sense. 
Going to our proof, a key step is the construction of the ``vertex operators'' which have 
in a sense maximal correlation across the separating corridor between $A_n$ and $B_n$ as stated in lemma \ref{lem:1}. 
To simplify, let us take  half lines $A_n, B_n$ separated by a corridor of width $2/n$ symmetrically around the origin. 
Proceeding somewhat informally to simplify the discussion, 
we consider instead the isometric vertex operators $V_n = u_{C_n B_n} v_{B_n}/{\sqrt{d}}$ where $C_n=(1/n,2/n)$
and $u_{C_n B_n}$ is a unitary charge transporter from $B_n$ to $C_n$. Then $V_n$ is localized in $(1/n,2/n)$ and it creates an incoherent superposition
of all irreducible charges in this interval by the Q-system construction, see app. \ref{sec:app}. 
Letting $J=J_{\cF}$ be the modular conjugation associated with the half-line $(0,\infty)$, we 
can say that $\bar V_n=J V_n J$ creates the opposite charges in the opposite interval $(-2/n,-1/n)$
because $J$ is basically the PCT operator exchanging $A_n$ with $B_n$,
and particle with anti-particle (Bisognano-Wichmann). 

Thus, the correlation which we want to maximize similar to lemma \ref{lem:1}
is 
\ben
1 \ge \langle \Omega | \bar V_n V_n^{} \Omega \rangle = \langle \Omega | V_{n}^{} \Delta^{1/2} V_{n}^* \Omega \rangle, 
\een
where the inequality is simply the Cauchy Schwarz inequality. The modular flow $\Delta^{it}$ corresponds to dilations by $e^t$ (Bisognano-Wichmann), and 
$V_n |\Omega \rangle$ must be approximately dilation invariant moving ever closer to the edge of $B_n$ when $n \to \infty$. Thus, 
the limit of $\langle \Omega | \bar V_n V_n^{} \Omega \rangle$ should indeed be $1$. Arguing just as in lemma \ref{lem:1}, one can also see 
at least formally that $\langle \psi_n | \bar V_n V_n^{} \psi_n \rangle$ should tend to $\lambda^{-1}$. 

Thus, in this sense, the quantity $S(\omega_\Omega | \omega_\Omega \circ E_n)$ is dominated in the limit $n \to \infty$
by particle anti-particle pair correlations very close to the edges across the corridor in accordance with the 
intuitive picture proposed by \cite{casini}.

\appendix

\section{Q-systems, subfactors and OPE \cite{longohopf,longo3,bischoff}}\label{sec:app}

A Q-system is a way to encode an inclusion of properly infinite von Neumann factors $\cN \subset \cM$
possessing a minimal conditional expectation $E:\cM \to \cN$ such that the index, denoted here by $d^2$, is finite. 
An important point is that the data in the Q-system only refer to the smaller factor, $\cN$.

Central to the construction is the notion of an endomorphism of $\cN$, which 
is an ultra-weakly continuous $*$-homomorphism such that $\theta(1)=1$.
Given two endomorphisms $\rho, \theta$, one says that a linear operator $T \in \Hom(\rho,\theta)$ (``intertwiner'') if 
$T\rho(n)=\theta(n) T$ for all $n \in \cN$. Two endomorphisms are called equivalent if there is a unitary 
intertwiner and irreducible if there is no non-trivial self-intertwiner.
One writes $\theta \cong \oplus_i \rho_i$ if there is 
a finite set of irreducible and mutually inequivalent 
endomorphisms $\rho_i$ and isometries $w_i \in \Hom(\rho_i, \theta)$ 
such that $\theta(n) = \sum w_i^* \rho_i(n) w_i$ for all $n \in \cN$ and such that $w_i w_j^* = \delta_{ij}1$.

\begin{definition}\label{Qsysdef}
A Q-system is a triple $(\theta, x, w)$ where: $\theta \cong \oplus_i \rho_i$ is an endomorphism of $\cN$, 
$w \in \Hom(\theta, id) \cap \cN$ and $x \in \Hom(\theta^2, \theta) \cap \cN$ such that
\ben
w^* x = \theta(w^*)x = 1, \quad 
x^2 = \theta(x)x, \quad
\theta(x^*)x = xx^* = x^* \theta(x), 
\een
as well as 
\ben 
w^* w = d \cdot 1, \quad 
x^* x =  d \cdot 1.
\een
\end{definition}

Given a Q-system, one defines an extension $\cM$  as follows. As a set, $\cM$ consists of 
all symbols of the form $nv$, where $n \in \cN$ with the product, $*$-operation, and unit defined by, respectively
\ben
n_1 v n_2 v = n_1 \theta(n_2) xv, \quad (nv)^* = w^* x^* \theta(n^*) v, \quad 1 = w^* v. 
\een
Associativity and consistency with the $*$-operation follow from the defining relations. 
The conditional expectation is related to the data by $E(nv)=d^{-1} nw$ and is used to induce the operator norm on $\cM$.
Conversely, given an inclusion of infinite (type III) factors $\cN \subset \cM$, the data of the Q-system and $v \in \cM$ can be found
by a canonical procedure and $d^2 = [\cM:\cN]$. 

Let $\rho_i$ and $w_i \in \Hom(\rho_i, \theta)$ be the endomorphisms and intertwiners corresponding to
the decomposition $\theta \cong \oplus_i \rho_i$ into irreducibles. Next, define
\ben
\psi_i = w_i^* v. 
\een 
The relations in def. \ref{Qsysdef} imply that the set of ``subsectors'' $\rho_i$ of $\theta$ is closed under composition (``fusion''), and that the following relations hold. 
Define: 
\ben
c_{i,j}^k = w_i^* \theta(w_j^*) x w_k, \quad w_{0} = w, 
\een
and let $\rho_0=id$ be the trivial endomorphism of $\cN$.
Then
\begin{itemize}
\item (Operator product expansion):
$
\psi_i \psi_j = \sum_k c_{i,j}^k \psi_k.
$

\item ($*$-operation)
$\psi_k^* = c_{\bar k, k}^0{}^* \psi_{\bar k}$ and 
$c_{j,k}^0 = \delta_{j, \bar k} R_k$, 
where $R_k \in \Hom(\rho_0, \bar \rho_k \rho_k)$ is the intertwiner characterizing the ``conjugate sector''. 

\item (Unit)
$
\psi_0 = 1.
$
\end{itemize} 

In the QFT context, one not only has one inclusion, but a net of inclusions $\cA(I) \subset \cF(I)$ \cite{longo3}. Furthermore, $\cA(I)$ is often taken to be the 
algebra generated by the smeared stress tensor inside $I$ (``Virasoro-net'').
From this, one should be able to construct an operator product 
expansion in the usual sense in the physics literature, although to establish the connection in full precision/generality 
remains an open problem. 

The basic idea is to consider the ``fields'' $\psi_{i,I}$ for each interval $I$. To obtain a pointlike 
vertex operator, we should shrink $I \to \{x\}$ while at the same time subtracting the vacuum expectation value $\langle \Omega |\psi_{I,i} \Omega \rangle 1$ 
and rescaling\footnote{$h_i$ is expected to equal the highest weight provided by the irreducible Virasoro representation $\rho_i$.} 
by $|I|^{-h_i}$ to obtain a finite limit, $V_i(x)$. These ``primary'' fields obey an OPE
with ``coefficients" $c_{i,j}^k(x,y)$ that are still operators in the Virasoro net. We should think of them as 
operator valued functions $c_{i,j}^k(x,y)=c_{i,j}^k(x,y, \{ L_n \})$. When formally 
expanded out as a power series in the Virasoro generators 
$\{ L_n \}_{n \in {\mathbb Z}}$, this ought to give the operator product expansion 
with certain numerical coefficients  containing on the right 
side the primary vertex operators $V_k(y)$ as well as their descendants 
$\phi_{k,\{n\}}(y) = [L_{n_1},[ \cdots L_{n_m},V_k(y)]]$, where $n_1 < n_2 < \dots < 0$. This 
is the form of the operator product expansion usually given in the physics literature. 
Representation theoretic considerations then formally 
determine the scaling of the numerical OPE coefficients. Such partly heuristic claims are at the basis of 
our discussion in sec. \ref{sec:disc}.

{\bf Acknowledgements:} SH\ is grateful to the Max-Planck Society for supporting the collaboration between MPI-MiS and Leipzig U., grant Proj.~Bez.\ M.FE.A.MATN0003. SH benefited from the KITP program ``Gravitational Holography''. This research was supported in part by the National Science Foundation under Grant No. NSF PHY-1748958.

\end{document}